\documentclass{article}[11pt]
\usepackage{geometry} 
\usepackage{graphicx}
\usepackage{amsfonts}
\usepackage{amssymb}
\usepackage{amsmath,amsthm,amsfonts}
 \usepackage{relsize}
\usepackage{wrapfig}
\usepackage{frame,color}
\usepackage{hyperref}
\usepackage{xspace}
\usepackage{framed}

\newtheorem{remark}{Remark}
\newtheorem{theorem}{Theorem}
\newtheorem{lemma}{Lemma}
\newtheorem{claim}{Claim}

\newtheorem{corollary}{Corollary}

\newcommand{\Rand}{\mathsf{Rand}}
\newcommand{\Det}{\mathsf{Det}}
\newcommand{\RandLOCAL}{\mathsf{RandLOCAL}}
\newcommand{\DetLOCAL}{\mathsf{DetLOCAL}}
\newcommand{\ID}{\operatorname{ID}}
\newcommand{\LOCAL}{\mathsf{LOCAL}}

\newcommand{\ignore}[1]{}

\newcommand{\bydef}{\stackrel{\rm def}{=}}
\newcommand{\paren}[1]{\mathopen{}\left( #1 \right)\mathclose{}}
\newcommand{\angbrack}[1]{\left< #1 \right>}

\newcommand{\f}[2]{\frac{#1}{#2}}

\newcommand{\poly}{{\operatorname{poly}}}

\newcommand{\dist}{\operatorname{dist}}
\newcommand{\Color}{\operatorname{color}}

\newcommand{\istrut}[2][0]{\rule[- #1 mm]{0mm}{#1 mm}\rule{0mm}{#2 mm}}

\newcommand{\Lovasz}{Lov\'{a}sz}

\title{An Exponential Separation Between Randomized\\
and Deterministic Complexity in the $\LOCAL$ Model\thanks{This work is supported by NSF grants CCF-1217338, CNS-1318294, and CCF-1514383.  Email: \{cyijun,tsvi,pettie\}@umich.edu}.}

\author{Yi-Jun Chang\\ University of Michigan \and
Tsvi Kopelowitz\\ University of Michigan \and
Seth Pettie\\ University of Michigan}

\begin{document}
\date{}

\maketitle

\begin{abstract}
Over the past 30 years numerous algorithms have been designed for symmetry breaking problems in the $\LOCAL$ model,
such as maximal matching, MIS, vertex coloring, and edge-coloring.  For most problems the best randomized algorithm
is at least exponentially faster than the best deterministic algorithm.  In this paper we prove that these exponential gaps are
{\em necessary}
and establish numerous connections between the deterministic and randomized complexities
in the $\LOCAL$ model.  Each of our results has a very compelling take-away message:
\begin{description}
\item[Fast $\Delta$-coloring of trees requires random bits.]
Building on the recent randomized lower bounds of Brandt et al.~\cite{BrandtEtal16}, we prove that the randomized complexity
of $\Delta$-coloring a tree with maximum degree $\Delta$ is $\Theta(\log_\Delta \log n)$, for any $\Delta \ge 55$,
whereas its deterministic complexity is $\Theta(\log_\Delta n)$ for any $\Delta\ge 3$.\footnote{For simplicity, here we suppress any $\log^\ast n$ additive term in $\Theta(\cdot)$.}  This also establishes a large
separation between the deterministic complexity of $\Delta$-coloring and $(\Delta+1)$-coloring trees.
\item[Randomized lower bounds imply deterministic lower bounds.]
We prove that any deterministic algorithm for a natural class of problems that runs in $O(1) + o(\log_\Delta n)$ rounds can be
transformed to run in $O(\log^* n - \log^*\Delta + 1)$ rounds.  If the transformed algorithm violates a lower bound (even allowing randomization), then one can conclude that the problem requires $\Omega(\log_\Delta n)$ time deterministically.
(This gives an alternate proof that deterministically
$\Delta$-coloring a tree with small $\Delta$ takes $\Omega(\log_\Delta n)$ rounds.)
\item[Deterministic lower bounds imply randomized lower bounds.] We prove
that the randomized complexity of any natural problem on instances of size $n$ is at least its deterministic complexity on instances of
size $\sqrt{\log n}$.  This shows that a {\em deterministic} $\Omega(\log_\Delta n)$ lower bound for any problem
($\Delta$-coloring a tree, for example) implies a {\em randomized} $\Omega(\log_\Delta \log n)$ lower bound.
It also illustrates that the {\em graph shattering} technique employed in recent randomized symmetry breaking algorithms
is {\em absolutely essential} to the $\LOCAL$ model.  For example,
it is provably impossible to improve the $2^{O(\sqrt{\log\log n})}$ terms in the complexities of the best MIS
and $(\Delta+1)$-coloring algorithms without {\em also} improving the $2^{O(\sqrt{\log n})}$-round Panconesi-Srinivasan algorithms.
\end{description}

\end{abstract}

\newpage

\section{Introduction}\label{sect:intro}
One of the central problems of theoretical computer science is to determine the value of {\em random bits}.
If the distinction is between computable vs. incomputable functions, random bits are provably useless in centralized models (Turing machines)~\cite{Sipser12}.
However, this is not true in the distributed world!
The celebrated Fischer-Lynch-Patterson theorem~\cite{FischerLP85} states that asynchronous deterministic agreement is impossible with one unannounced failure,
yet it is possible to accomplish with probability 1 using randomization.
See Ben-Or~\cite{Ben-Or83} and~\cite{Bracha84,Rabin83,KingS13}.
There are also a number of basic symmetry breaking tasks that are trivially impossible
to solve by identical, synchronized, deterministic processes, for example, medium access control to an Ethernet-like channel.

In this paper we examine the value of random bits in Linial's~\cite{Linial92} $\LOCAL$ model, which, for the sake of clarity, we
bifurcate into two models $\RandLOCAL$ and $\DetLOCAL$.  In both models the input graph $G=(V,E)$ and communications network are identical.
Each vertex hosts a processor and all vertices run the same algorithm.  Each edge supports communication in both directions.
The computation proceeds in synchronized {\em rounds}.  In a round, each processor
performs some computation and sends a message along each incident edge, which is delivered before the beginning
of the next round.  Each vertex $v$ is initially aware of its degree $\deg(v)$ and certain global parameters such as $n \bydef |V|$,
$\Delta = \Delta(G) \bydef \max_{v\in V} \deg(v)$, and possibly others.\footnote{The assumption that global parameters are common knowledge can sometimes
be removed; see Korman, Sereni, and Viennot~\cite{KormanSV13}.} In the $\LOCAL$ model the only measure of efficiency
is the number of rounds.  {\em All local computation is free and the size of messages is unbounded.}
Henceforth ``time'' refers to the number of rounds.

\begin{description}
\item[$\DetLOCAL$:] In order to avoid trivial impossibilities, all vertices are assumed to hold unique $\Theta(\log n)$-bit IDs.
Except for the registers holding $\deg(v)$ and $\ID(v)$, the initial state of $v$ is identical to every other vertex.
The algorithm executed at each vertex is deterministic.

\item[$\RandLOCAL$:] In this model each vertex may locally generate an unbounded number of independent truly random bits.
(There are no globally shared random bits.)  Except for the register holding $\deg(v)$, the initial state of $v$ is identical to
every other vertex.  Algorithms in this model operate for a specified number of rounds and have some probability of {\em failure},
the definition of which is problem specific.
We usually only consider algorithms whose global probability of failure is at most $1/\poly(n)$.
\end{description}

Observe that the lack of IDs in $\RandLOCAL$ is not a practical limitation.  Before the first round each vertex can locally
generate a random $\Theta(\log n)$-bit ID, which is unique with probability $1 - 1/\poly(n)$. For technical reasons it is convenient
to assume that vertices are not initially differentiated by IDs.\footnote{Notice that the role of ``$n$'' is different in the two $\LOCAL$ models:
in $\DetLOCAL$ it only affects the ID length whereas in $\RandLOCAL$ it only affects the failure probability.}

Early work in the $\LOCAL$ models suggested that randomness is of limited help.  Naor~\cite{Naor91} showed that Linial's $\Omega(\log^* n)$ lower bound~\cite{Linial92} for 3-coloring the ring
holds even in $\RandLOCAL$, and Naor and Stockmeyer~\cite{NaorS95} proved that the class of problems solvable by $O(1)$-round algorithms is the same in $\RandLOCAL$ and $\DetLOCAL$.
However, in the intervening decades we have seen dozens of examples of symmetry breaking algorithms for $\RandLOCAL$ that are substantially faster than their
counterparts in $\DetLOCAL$; see~\cite{BEPS16} for an extensive survey or Table~\ref{table:symmetry-breaking} for a glimpse at three archetypal problems:
maximal independent set (MIS), maximal matching, and $(\Delta+1)$-coloring.
\begin{table}[h!]
\centering
{\small
\begin{tabular}{|l|ll|l|}
\multicolumn{1}{l}{\bf Problem}				&	\multicolumn{2}{l}{\bf Model and Result}						& \multicolumn{1}{l}{\bf Citation}\\\hline
{\bf MIS:}				&	$\DetLOCAL$:		& $O\paren{\min\left\{\Delta + \log^*n, \; 2^{O(\sqrt{\log n})}\right\}}$				& \cite{BarenboimEK14,PanconesiS96}\\
					& 	$\RandLOCAL$:	& $O\paren{\log \Delta + 2^{O(\sqrt{\log\log n})}}$					& \cite{Ghaffari16}\\
\istrut[5]{0}			&	Lower Bound:		& $\Omega\paren{\min\left\{\sqrt{\log n/\log \log n}, \; \log\Delta/\log\log\Delta + \log^* n\right\}}$	& \cite{KuhnMW04,Linial92,Naor91}\\
{\bf Maximal Matching:}\istrut[2.5]{0}	&	$\DetLOCAL$:		& $O\paren{\min\left\{\Delta + \log^* n, \; \log^4 n\right\}}$			\istrut[2.5]{0}			& \cite{PanconesiR01,HanckowiakKP01}\\
					& 	$\RandLOCAL$:	& $O\paren{\log \Delta + \log^4\log n}$							& \cite{BEPS16}\\
\istrut[5]{0}			& 	Lower Bound:		& $\Omega\paren{\min\left\{\sqrt{\log n/\log \log n}, \; \log\Delta/\log\log\Delta + \log^* n\right\}}$			& \cite{KuhnMW04,Linial92,Naor91}\\
{\bf $(\Delta+1)$-Coloring}&	$\DetLOCAL$:		& $O\paren{\min\left\{\Delta^{3/4}\log\Delta + \log^* n, \; 2^{O(\sqrt{\log n})}\right\}}$	& \cite{Barenboim15,PanconesiS96}\\
					&	$\RandLOCAL$:	& $O\paren{\sqrt{\log\Delta} + 2^{O(\sqrt{\log\log n})}}$				& \cite{HarrisSS16}\\
					&	Lower  Bound:		& $\Omega\paren{\log^* n}$									& \cite{Linial92,Naor91}\\\hline
\end{tabular}
}
\caption{\label{table:symmetry-breaking}A sample of symmetry breaking results for three problems.}
\end{table}

\paragraph{Graph Shattering.}
A little pattern matching in the bounds of Table~\ref{table:symmetry-breaking} shows that the randomized symmetry breaking algorithms are exponentially faster {\em in two ways}.
Their dependence on $\Delta$ is exponentially faster and their dependence on $n$ is usually identical to the best deterministic
complexity, but for $\poly(\log n)$-size instances, for example, $2^{O(\sqrt{\log n})}$ becomes $2^{O(\sqrt{\log\log n})}$.
This second phenomenon is no coincidence!  It is a direct result of the {\em graph shattering} approach to symmetry breaking
used in~\cite{BEPS16} and further in~\cite{ChungPS14,ElkinPS15,Ghaffari16,HarrisSS16,BishtKP14,KothapalliP12,PettieS15}.  The idea is to apply some
randomized procedure that fixes some fragment of the output (e.g., part of the MIS is fixed, part of the coloring is fixed, etc.),
thereby effectively removing a large fraction of the vertices from further consideration.
If it can be shown that the connected components in the subgraph still under consideration have size $\poly(\log n)$,
one can revert to the best available {\em deterministic} algorithm and solve the problem on each component of the ``shattered'' graph in parallel.

\paragraph{Lower Bounds in the $\LOCAL$ Model.}
Until recently, the main principle used to prove lower bounds in the $\LOCAL$ model was {\em indistinguishability}.
The first application of this principle was by Linial~\cite{Linial92} himself, who argued that any algorithm for coloring degree-$\Delta$ trees
either uses $\Omega(\Delta/\log \Delta)$ colors or takes $\Omega(\log_\Delta n)$ time.  The proof is as follows
(i) in $o(\log_\Delta n)$ time, a vertex {\em cannot always distinguish} whether the input graph
$G$ is a tree or a graph with girth $\Omega(\log_\Delta n)$,
(ii) for all $\Delta$ and all $n$, there exists a degree-$\Delta$ graph with girth $\Omega(\log_\Delta n)$ and chromatic number $\chi = \Omega(\Delta/\log\Delta)$, hence\footnote{Linial~\cite{Linial92} actually only used the existence of $\Delta$-regular graphs with high girth and chromatic number $\Omega(\sqrt{\Delta})$.  See~\cite{Bollobas78b} for constructions with chromatic number $\Omega(\Delta/\log \Delta)$.}
(iii) any $o(\log_\Delta n)$-time algorithm for coloring trees could also color
such a graph, and therefore must use at least $\chi$ colors.  A significantly more subtle indistinguishability argument was
used by Kuhn, Moscibroda, and Wattenhofer~\cite{KuhnMW04}, who showed that $O(1)$-approximate vertex cover, maximal matching,
MIS, and several other problems have $\Omega(\min\{\log\Delta/\log\log\Delta, \sqrt{\log n/\log \log n}\})$ lower bounds\footnote{In~\cite{KuhnMW10} the same authors argue that these problems have a lower bound of $\Omega(\min\{\log\Delta, \sqrt{\log n}\})$. However, recently Bar-Yehuda, Censor-Hillel, and Schwartzman~\cite{BarYehudaCS16} pointed our an error in their proof.}.
By its nature, indistinguishability is not very good at separating randomized and deterministic complexities.
Very recently, Brandt et al.~\cite{BrandtEtal16} developed a lower bound technique that explicitly incorporates error probabilities
and proved that several problems on graphs with constant 
$\Delta$ take $\Omega(\log\log n)$ time in $\RandLOCAL$ (with error probability $1/\poly(n)$)
such as {\em sinkless orientation}, {\em sinkless coloring}, and $\Delta$-coloring.
Refer to Section~\ref{sec.prelim} for definitions of these problems.
Since the existence of a sinkless orientation can be proved with the \Lovasz{} local lemma (LLL), this
gave $\Omega(\log\log n)$ lower bounds on distributed algorithms for the constructive LLL.
See~\cite{ChungPS14,Ghaffari16} for upper bounds on the distributed LLL.

\subsection{New Results}

In this paper we exhibit an exponential separation between $\RandLOCAL$ and $\DetLOCAL$
for several {\em specific} symmetry breaking problems.  More generally, we give new connections between
the randomized and deterministic complexities of all {\em locally checkable labeling} problems
(refer to Section~\ref{sec.prelim} for a definition of LCLs),
a class that includes essentially any natural symmetry breaking problem.

\begin{description}
\item[Separation of $\RandLOCAL$ and $\DetLOCAL$.]
We extend Brandt et al.'s~\cite{BrandtEtal16}
randomized lower bound as follows: on degree-$\Delta$ graphs, $\Delta$-coloring takes $\Omega(\log_\Delta \log n)$ time in $\RandLOCAL$ and $\Omega(\log_\Delta n)$ time in $\DetLOCAL$.
The hard graphs in this lower bound have girth $\Omega(\log_\Delta n)$, so by the indistinguishability principle, these lower bounds also apply to $\Delta$-coloring trees.
On the upper bound side, Barenboim and Elkin~\cite{BarenboimE10} showed that for $\Delta\ge 3$, $\Delta$-coloring trees takes $O(\log_\Delta n + \log^* n)$ time in $\DetLOCAL$.
We give an elementary proof that for $\Delta \ge 55$,
$\Delta$-coloring trees can be done in $O(\log_\Delta \log n + \log^* n)$ time in $\RandLOCAL$, matching Brandt et al.'s~\cite{BrandtEtal16} lower bound up to the $\log^* n$.
A more complicated algorithm for $\Delta$-coloring trees could be derived from~\cite{PettieS15}, for $\Delta > \Delta_0$ and
some very large constant $\Delta_0$.\footnote{The reason we are interested in minimizing the $\Delta_0 \leq \Delta$ 
for which the algorithm
works is somewhat technical.  It seems as if $\Delta$-coloring trees is a problem whose character makes 
a {\em qualitative} transition when $\Delta$ is a small enough constant.  Using our technique (graph shattering)
we may be able to replace 55 with a smaller constant, 
{\em but not too small}.  
Any algorithm that 3-colors 3-regular trees, for example, will need to be qualitatively very different in its design.}

\item[Randomized lower bounds imply deterministic lower bounds.]
We give a second, more generic proof that $\Delta$-coloring trees takes $\Omega(\log_\Delta n)$ time.
The proof shows that any $f(\Delta) + o(\log_\Delta n)$ time algorithm for an LCL problem
can be transformed in a black box way to run in $O((1+f(\Delta))(\log^* n - \log^* \Delta + 1))$ time.
Thus, on bounded-degree graphs, there are
no ``natural'' deterministic time bounds between $\omega(\log^* n)$ and $o(\log n)$.  Any $\omega(\log^* n)$ lower bound for
bounded degree graphs (in either $\RandLOCAL$ or $\DetLOCAL$) {\em immediately} implies an $\Omega(\log n)$
lower bound in $\DetLOCAL$.  This reduction can be parameterized in many different ways.
Under a different parametrization it shows that any $O(\log^{1-\f{1}{k+1}} n)$-time $\DetLOCAL$ algorithm for an LCL
problem can be transformed to run in $O(\log^k \Delta(\log^* n - \log^* \Delta + 1))$-time.
For example, if one were to develop a deterministic $O(\sqrt{\log n})$-time MIS or maximal matching algorithm---almost matching one of the KMW~\cite{KuhnMW04} lower bounds---it would immediately imply an
$O(\log\Delta(\log^* n - \log^* \Delta +1))$-time MIS/maximal matching algorithm, which almost matches the {\em other} KMW lower bound, for $\Delta > \log^{O(1)} n$.
By some strange coincidence, \cite{BEPS16} gave an analogous reduction for MIS/maximal matching in bounded arboricity graphs, but for $\RandLOCAL$ and in the {\em reverse direction}.  Specifically, any
$O(\log^k \Delta \,+\, f(n))$-time $\RandLOCAL$ MIS/maximal matching algorithm can be transformed into an $O(\log^{1-\f{1}{k+1}} n \,+\, f(n))$-time $\RandLOCAL$ algorithm for bounded arboricity graphs.

\item[Deterministic lower bounds imply randomized lower bounds.]
We prove that for any LCL problem, its $\RandLOCAL$ complexity on instances of size $n$
is at least its $\DetLOCAL$ complexity on instances of size $\sqrt{\log n}$.
This {\em reverses} the implication proved above.  For example, if we begin with a proof that $\Delta$-coloring takes
$\Omega(\log_\Delta n)$ time in $\DetLOCAL$,
then we conclude that it must take $\Omega(\log_\Delta\log n)$ time in $\RandLOCAL$.
This result has a very clear take-away message: the {\em graph shattering} technique applied by recent randomized symmetry breaking
algorithms~\cite{BEPS16,Ghaffari16} is {\em inherent} to the $\RandLOCAL$ model and every optimal $\RandLOCAL$ algorithm for instances
of size $n$ must, in some way, encode an optimal $\DetLOCAL$ algorithm on $\poly(\log n)$-size instances.
It is therefore impossible to improve the $2^{O(\sqrt{\log\log n})}$ terms in
the $\RandLOCAL$ MIS and coloring algorithms
of~\cite{BEPS16,Ghaffari16,HarrisSS16,ElkinPS15}
without also improving the $2^{O(\sqrt{\log n})}$-time $\DetLOCAL$ algorithms of Panconesi and Srinivasan~\cite{PanconesiS96},
and it is impossible to improve the $O(\log^4\log n)$ term in the $\RandLOCAL$ maximal matching algorithm of~\cite{BEPS16}
without also improving the $O(\log^4 n)$ $\DetLOCAL$ maximal matching algorithm of~\cite{HanckowiakKP01}.
\end{description}

\section{Preliminaries \label{sec.prelim}}

\paragraph{Graph Notation.} For a graph $G=(V,E)$ and for $u,v\in V$, let $\dist_G(v,u)$ be the distance between $v$ and $u$ in $G$.
Let $N(v) = \{u \;|\; (v,u) \in E\}$ be the neighborhood of $v$ and let
$N^r(v) = \{u \;|\; \dist_G(v,u) \le r\}$ be the set of all vertices within distance $r$ of $v$.

\paragraph{Locally Checkable Labeling.}
The class of \emph{Locally Checkable Labeling} (LCL)~\cite{NaorS95} problems are intuitively those graph problems
whose solutions can be verified in $O(1)$ rounds, given a suitable labeling of the graph.
Formally, an LCL problem is defined by a fixed radius $r$, a finite set $\Sigma$ of vertex labels,
and a set $\mathcal{C}$ of acceptable labeled subgraphs.
For any legal solution $I$ to the problem there is a labeling $\lambda_I \;:\; V\rightarrow \Sigma$
that encodes $I$ (plus possibly other information) such that for each $v\in V$,
the labeled subgraph induced by $N^r(v)$ lies in $\mathcal{C}$.
Moreover, for any {\em non-solution} $I'$ to the problem, there is no labeling $\lambda_{I'}$ with this property.
The following symmetry breaking problems are LCLs for $r=1$.

\begin{itemize}
  \item \noindent{\sc Maximal Independent Set (MIS)}. Given a graph $G=(V,E)$, find a set $I \subseteq V$ such that  for any vertex $v \in V$, we have $N(v) \cap I = \emptyset$ iff  $v \in I$.

  \item \noindent $k$-{\sc Coloring}. Given a graph $G=(V,E)$, find an assignment $V \rightarrow \{1,2, \ldots, k\}$ such that  for each edge $\{u,v\} \in E$, $u$ and $v$ are assigned to different numbers (also called colors).

\end{itemize}

For MIS it suffices to label vertices with $\Sigma = \{0,1\}$ indicating whether they are in the MIS.
For $k$-Coloring we use $\Sigma = \{1,\ldots,k\}$.  The definition of LCLs is easily generalized to the case
where the input graph $G$ is supplemented with some labeling (e.g., an edge-coloring) or where $\lambda$
labels both vertices and edges.  Brandt et al.~\cite{BrandtEtal16} considered the following problems.

\begin{itemize}
  \item \noindent $\Delta$-{\sc Sinkless Coloring}. Given a $\Delta$-regular graph $G=(V,E)$ and a proper $\Delta$-edge coloring of $E$ using colors in $\{1,2, \ldots, \Delta\}$, find a $\Delta$-coloring of $V$ using colors in $\{1,2, \ldots, \Delta\}$ such that  there is no edge $\{u,v\} \in E$ for which $u$, $v$ and $\{u,v\}$ all have the same color.

  \item \noindent $\Delta$-{\sc Sinkless Orientation}. Given a $\Delta$-regular graph $G=(V,E)$ and a proper $\Delta$-edge coloring of $E$, find an orientation of the edges such that all edges have out-degree $\geq 1$.
\end{itemize}

Observe that both $\Delta$-Sinkless Coloring and $\Delta$-Sinkless Orientation are LCL graph problems with $r=1$.  For Sinkless Orientation $\Sigma = \{\rightarrow,\leftarrow\}^\Delta$ encodes the directions of all edges incident to a vertex, and the radius $r=1$ is necessary and sufficient to verify that the orientations declared by both endpoints of an edge are consistent.

\paragraph{Linial's coloring.} In the $\DetLOCAL$ model
the initial $\Theta(\log n)$-bit IDs can be viewed as an $n^{O(1)}$-coloring of the graph.
Our algorithms make frequent use of Linial's~\cite{Linial92} coloring algorithm, which recolors the vertices using a smaller palette.

\begin{theorem} [\cite{Linial92}] \label{thm:reduction}
Let $G$ be a graph which has been $k$-colored. Then it is possible to deterministically re-color $G$ using $5\Delta^2 \log k$ colors in one round.
\end{theorem}

\begin{theorem} [\cite{Linial92}] \label{cor:linial-coloring}
There exists a universal constant $\beta >0$ such that there is a $\DetLOCAL$ algorithm that computes a $\beta \cdot \Delta^2$-coloring of a graph in $\mathcal{O}(\log^\ast n - \log^\ast \Delta + 1)$ time.
\end{theorem}

\section{The Necessity of Graph Shattering}\label{sect:graph-shattering}

Theorem~\ref{thm:graph-shattering} establishes that the graph shattering technique~\cite{BEPS16} is optimal and unavoidable
in $\RandLOCAL$.  In particular, the randomized complexity of any symmetry breaking problem always hinges on its deterministic complexity.

\begin{theorem}\label{thm:graph-shattering}
Let $\mathcal{P}$ be an LCL problem.
Define $\Det_{\mathcal{P}}(n,\Delta)$ to be the complexity of the optimal deterministic
algorithm for $\mathcal{P}$ in the $\DetLOCAL$ model
and define
$\Rand_{\mathcal{P}}(n,\Delta)$
to be its complexity in the $\RandLOCAL$ model,
with global error probability $1/n$.
Then
\[
\Det_{\mathcal{P}}(n,\Delta) \;\le\; \Rand_{\mathcal{P}}(2^{n^2},\Delta).
\]
\end{theorem}

\begin{proof}
Let $\mathcal{A}_{\Rand}$ be a randomized algorithm for $\mathcal{P}$.  Each vertex running
$\mathcal{A}_{\Rand}$ generates a string of $r(n,\Delta)$ random bits and
proceeds for $t(n,\Delta)$ rounds, where $r$ and $t$ are two arbitrary functions.
The probability that the algorithm fails in any way is at most $1/n$.
Our goal is to convert $\mathcal{A}_{\Rand}$ into a deterministic algorithm
$\mathcal{A}_{\Det}$ in the $\DetLOCAL$ model.  Let $G=(V,E)$ be the network on
which $\mathcal{A}_{\Det}$ runs.  Initially each $v\in V$ knows $n=|V|,\Delta,$ and
a unique $\ID(v) \in \{0,1\}^{c\log n}$.  Let $\mathcal{G}_{n,\Delta}$ be the set of {\em all}
$n$-vertex graphs with unique vertex IDs in $\{0,1\}^{c\log n}$ and maximum degree at most $\Delta$.
Regardless of $\Delta$,
\[
\left|\mathcal{G}_{n,\Delta}\right| \le 2^{{n\choose 2} + cn\log n} \ll 2^{n^2} \bydef N
\]
Imagine simulating $\mathcal{A}_{\Rand}$ on a graph $G' \in \mathcal{G}_{n,\Delta}$
whose vertices are given input parameters $(N,\Delta)$, that is, we imagine $G'$ is disconnected
from the remaining $N-n$ vertices.  The probability that $\mathcal{A}_{\Rand}$ fails on an
$N$-vertex graph is at most $1/N$,
so the probability that any vertex in $G'$ witnesses a failure is also certainly at most $1/N$.

Suppose we select a function $\phi : \{0,1\}^{c\log n} \rightarrow \{0,1\}^{r(N,\Delta)}$ uniformly at random
from the space of all such functions.  Define $\mathcal{A}_{\Det}[\phi]$ to be the {\em deterministic} algorithm
that simulates $\mathcal{A}_{\Rand}$ for $t(N,\Delta)$ steps, where the string of random bits generated
by $v$ is fixed to be $\phi(\ID(v))$.  We shall call $\phi$ a {\em bad} function if $\mathcal{A}_{\Det}[\phi]$
fails to compute the correct answer on some member of $\mathcal{G}_{n,\Delta}$.  By the union bound,
\begin{align*}
\Pr_\phi(\mbox{$\phi$ is bad}) &\le \sum_{G' \in \mathcal{G}_{n,\Delta}} \Pr_\phi(\mbox{$\mathcal{A}_{\Det}[\phi]$ errs on $G'$})\\
					       &= \sum_{G' \in \mathcal{G}_{n,\Delta}} \Pr(\mbox{$\mathcal{A}_{\Rand}$ errs on $G'$, with input parameters $(N,\Delta)$})\\
					       &\le \left|\mathcal{G}_{n,\Delta}\right|/N \:<\: 1.
\end{align*}
Thus, there exists some good $\phi$.  Any $\phi$ can be encoded as a long bit-string
$\angbrack{\phi} \bydef \phi(0)\phi(1)\cdots\phi(2^{c\log n}-1)$.
Define $\phi^\star$ to be the good function for which $\angbrack{\phi^\star}$ is lexicographically first.

The algorithm $\mathcal{A}_{\Det}$ is as follows.  Each vertex $v$, given input parameters $(n,\Delta)$,
first computes $N = 2^{n^2}, t(N,\Delta), r(N,\Delta),$ then performs the simulations of $\mathcal{A}_{\Rand}$
necessary to compute $\phi^\star$.  Once $\phi^\star$ is computed it executes
$\mathcal{A}_{\Det}[\phi^\star]$ for $t(N,\Delta)$ rounds.
By definition, $\mathcal{A}_{\Det}[\phi^\star]$ never errs when run on any member of $\mathcal{G}_{n,\Delta}$.
\end{proof}

\begin{remark}
Theorem~\ref{thm:graph-shattering} works equally well when $t$ and $r$ are functions
of $n,\Delta,$ and possibly other quantitative global graph parameters.  For example, the time
may depend on measures of local sparsity (as in \cite{ElkinPS15}),
arboricity/degeneracy (as in~\cite{BarenboimE10,BEPS16}),
or neighborhood growth (as in~\cite{SchneiderW10-J}).
\end{remark}

Naor and Stockmeyer~\cite{NaorS95} proved that the class of truly local ($O(1)$-time) problems in $\RandLOCAL$ and $\DetLOCAL$ is
identical for bounded $\Delta$. Theorem~\ref{thm:graph-shattering} implies something slightly stronger, since $\log^* n$ and $\log^*(\sqrt{\log n})$ differ by a constant.

\begin{corollary}
Any $\RandLOCAL$ algorithm for an LCL running in $t(n) = 2^{O(\log^* n)}$ time can be derandomized
without asymptotic penalty.
The corresponding $\DetLOCAL$ algorithm runs in $O(t(n))$ time.
\end{corollary}

\section{Lower bounds for $\Delta$-coloring $\Delta$-regular Trees \label{sec.lb}}
In this section we prove that on degree-$\Delta$ graphs with girth $\Omega(\log_\Delta n)$, $\Delta$-coloring takes
$\Omega(\log_\Delta \log n)$ time in $\RandLOCAL$ and $\Omega(\log_\Delta n)$ time in $\DetLOCAL$.
Since the girth of the graphs used to prove these lower bounds is $\Omega(\log_\Delta n)$,
by the indistinguishability principle they also apply to the problem of $\Delta$-coloring trees.

\paragraph{Sinkless coloring and sinkless orientations.}
Brandt et.al.~\cite{BrandtEtal16} proved $\Omega (\log \log n)$ lower bounds on $\RandLOCAL$ algorithms, that have a $1/\poly(n)$ probability of failure, for sinkless coloring and sinkless orientation of 3-regular graphs.
We say that a sinkless coloring algorithm $\mathcal{A}$ has failure probability $p$ if, for {\em each} individual edge $e=\{u,v\}$,
the probability that $\Color(u)=\Color(v)=\Color(\{u,v\})$ is at most $p$.  Thus, by the union bound, the \emph{global} probability of failure is at most $p|E|$.
We say a that sinkless orientation algorithm $\mathcal{A}$ has failure probability $p$ if, for each $v\in V$, the probability that $v$ is a sink is at most $p$.
We say that monochromatic edges and sinks are {\em forbidden configurations} for sinkless coloring and sinkless orientation, respectively.

The following two lemmas are proven in~\cite{BrandtEtal16} for $\Delta=3$.
It is straightforward to go through the details of the proof and track the dependence on $\Delta$.

\begin{lemma} [\cite{BrandtEtal16}] \label{lem-speedup1} Let $G=(V,E,\psi)$ be a $\Delta$-regular graph with girth $g$ that is equipped with a proper $\Delta$-edge coloring $\psi$.
Suppose that there is a $\RandLOCAL$ algorithm $\mathcal{A}$ for $\Delta$-sinkless coloring taking $t < \frac{g-1}{2}$ rounds such that $\forall e \in E$,
$\mathcal{A}$ outputs  a forbidden configuration at $e$ with probability at most $p$.
Then there is a $\RandLOCAL$ algorithm $\mathcal{A}'$ for $\Delta$-sinkless orientation taking $t$ rounds such that $\forall v \in V$, $\mathcal{A}'$ outputs
a forbidden configuration at $v$ with probability at most $2\Delta p^{1/3}$.
\end{lemma}

\begin{lemma} [\cite{BrandtEtal16}] \label{lem-speedup2} Let $G=(V,E,\psi)$ be a $\Delta$-regular graph with girth $g$ that is equipped with a proper $\Delta$-edge coloring $\psi$.
Suppose that there is a $\RandLOCAL$ algorithm $\mathcal{A}'$ for sinkless orientation taking $t < \frac{g-1}{2}$ rounds such that $\forall v \in V$,
$\mathcal{A}'$ outputs a forbidden configuration at $v$  with probability at most $p$.
Then there is a $\RandLOCAL$ algorithm $\mathcal{A}$ for $\Delta$-sinkless coloring taking $t-1$ rounds such that $\forall e \in E$,
$\mathcal{A}$ outputs  a forbidden configuration at $e$ with  probability at most $4p^{1/(\Delta+1)}$.
\end{lemma}

The following theorem generalizes Corollary 25 in \cite{BrandtEtal16} to allow non-constant $\Delta$ and arbitrary failure probability $p$.

\begin{theorem} \label{thm:lb-rnd-generalized}
Any $\RandLOCAL$ algorithm for $\Delta$-coloring a graph with degree at most $\Delta$ and error probability $p$ takes
at least $t = \min\{\epsilon \log_{3(\Delta+1)} \ln(1/p), \; \epsilon \log_\Delta n\} \ge 1$ rounds for a sufficiently small $\epsilon>0$.
\end{theorem}

\begin{proof}
For any $\Delta\ge 3$ there exist a bipartite $\Delta$-regular graphs with girth $\Omega(\log_\Delta n)$; see~\cite{Dahan14,Bollobas78}.
Such graphs are trivially $\Delta$-edge colorable.  Moreover, any $\Delta$-coloring of such a graph is also a valid $\Delta$-sinkless coloring.
Applying Lemmas~\ref{lem-speedup1} and~\ref{lem-speedup2} we conclude that any $t$-round $\Delta$-sinkless coloring algorithm with error
probability $p$ can be transformed into a $(t-1)$-round $\Delta$-sinkless coloring algorithm with error probability
$4(2\Delta)^{\frac{1}{\Delta+1}}p^{\frac{1}{3(\Delta+1)}} < 7p^{\frac{1}{3(\Delta+1)}}$.  Iterating this process $t$ times,
it follows that there exists a 0-round $\Delta$-sinkless coloring algorithm with (local) error probability $O(p^{(\frac{1}{3(\Delta+1)})^t})$.
Note that
\[
p^{(\frac{1}{3(\Delta+1)})^t} \ge p^{(\frac{1}{3(\Delta+1)})^{\epsilon \log_{3(\Delta+1)} \ln (1/p)}} = p^{(\ln p)^{-\epsilon}} = \exp(-(\ln(1/p))^{1-\epsilon})
\]
Because the graph is $\Delta$-regular and the vertices undifferentiated by IDs, any 0-round $\RandLOCAL$ algorithm
colors each vertex independently according to the same distribution.
The probability that any vertex is involved in a forbidden configuration (a monochromatic edge)
is therefore at least $1/\Delta^2$.  Since $\epsilon\log_{3(\Delta+1)}\ln(1/p) \ge 1$ we have $\Delta < \ln(1/p)$, but

$$\frac{1}{\Delta^2} \geq  \exp(-2\ln\ln(1/p)) \gg \exp\left(-\left(\ln(1/p)\right)^{1-\epsilon}\right).$$

Thus, there is no $\RandLOCAL$ $\Delta$-sinkless coloring algorithm that takes $t$-rounds and errs with probability $p$.
\end{proof}

Corollary~\ref{thm:lb-rnd} is an immediate consequence of Theorem~\ref{thm:lb-rnd-generalized}.

\begin{corollary} \label{thm:lb-rnd}
Any $\RandLOCAL$ algorithm for $\Delta$-coloring a graph with global error probability $1/\poly(n)$ takes $\Omega(\log_\Delta\log n)$ time.
\end{corollary}

Theorem~\ref{thm:lb-rnd-generalized} does not immediately extend to $\DetLOCAL$.
Recall that in the $\DetLOCAL$ model vertices are initially endowed with $O(\log n)$-bit IDs whereas
in $\RandLOCAL$ they are undifferentiated.

\begin{theorem} \label{thm:lb-det}
Any $\DetLOCAL$ algorithm that $\Delta$-colors degree-$\Delta$ graphs with girth $\Omega(\log_\Delta n)$ or degree-$\Delta$ trees
requires $\Omega ({\log_\Delta n})$ time.
\end{theorem}
\begin{proof}
Let $\mathcal{A}_{\Det}$ be a $\DetLOCAL$ algorithm that $\Delta$-colors a graph in $t=t(n,\Delta)$ rounds and $G$ be the input graph.
We construct a $\RandLOCAL$ algorithm $\mathcal{A}_{\Rand}$
taking $O(t)$ rounds as follows.  Before the first round each vertex locally generates a random
$n$-bit ID.
Assume for the time being that these IDs are unique, and therefore constitute a $2^n$-coloring of $G$.
Let $G' = (V, \{(u,v) \;|\; \dist_G(u,v) \le 2t+1\})$.  The maximum degree $\Delta'$ in $G'$ is clearly less than $n$.
We apply one step of Linial's recoloring algorithm (Theorem~\ref{thm:reduction}) to $G'$ and obtain a coloring with palette size $O(\Delta'^2\log(2^n)) = O(n^3)$.
A step of Linial's algorithm in $G'$ is simulated in $G$ using $O(t)$ time.  Using these colors as $(3\log n+O(1))$-bit IDs,
we simulate $\mathcal{A}_{\Det}$ in $G$ for $t$ steps.  Since no vertex can see two vertices with the same ID, this algorithm
necessarily behaves as if all IDs are unique.  Observe that because $\mathcal{A}_{\Det}$ is deterministic,
the only way $\mathcal{A}_{\Rand}$ can err is if the initial $n$-bit IDs fail to be unique.  This occurs with probability $p < n^2/2^n$.
By Theorem~\ref{thm:lb-rnd-generalized} $\mathcal{A}_{\Rand}$ takes $\Omega(\min\{\log_\Delta\log(1/p), \, \log_\Delta n\}) = \Omega(\log_\Delta n)$ time.
\end{proof}

\section{Gaps in Deterministic Time Complexity \label{sec.lb2}}

The Time Hierarchy Theorem informally says that a Turing machine can solve more problems given more time. A similar question can be asked in the setting of distributed computation. For example, does increasing the number of rounds from $\Theta(\log^\ast n)$ to $\Theta(\log \log n)$ allow one to solve more problems? In this section, we will demonstrate a general technique that allows one to speedup deterministic algorithms in the $\DetLOCAL$ model. Based on this technique, we demonstrate the existence of a ``gap'' in possible $\DetLOCAL$ complexities.

A graph class is {\em hereditary} if it is closed under removing vertices and edges.
Examples of hereditary graph classes are general graphs, forests, bounded arboricity graphs, triangle-free graphs, planar graphs, and chordal graphs.
We prove that for graphs with constant $\Delta$ the time complexity of {\em any} LCL problem on any graph from a \emph{hereditary} graph class is either $\Omega(\log n)$ or $O(\log^\ast n)$.

\begin{theorem} \label{thm:lb-det-linial}
Let $\mathcal{P}$ be an LCL graph problem with parameters $r$, $\Sigma$, and $\mathcal{C}$,  and let $\mathcal{A}$ be a $\DetLOCAL$ algorithm for solving $\mathcal{P}$.
Let $\beta$ be the universal constant from Theorem~\ref{cor:linial-coloring}.
Suppose that the cost of $\mathcal{A}$ on instances of $\mathcal{P}$ with $n$ vertices, where the instance is taken from a hereditary graph class, is at most $ f(\Delta) + \epsilon \log_{\Delta} n $ time, where $f(\Delta)\geq 0$ and $\epsilon = \frac{1}{4+4\log\beta+4r}$ is a constant. Then there exists a $\DetLOCAL$ algorithm $\mathcal{A}'$ that solves $\mathcal{P}$ on the same instances in $O\left((1+f(\Delta)(\log^\ast n - \log^\ast \Delta+1) \right)$ time.
\end{theorem}
\begin{proof}
Notice that for any instance of $\mathcal{P}$ with $n$ vertices and ID length $\ell$, it must be that $\ell \geq \log n$ and so the running time of $\mathcal{A}$ on such instances is bounded by $T(\Delta, \ell) \leq f(\Delta) + \frac{\epsilon \ell}{\log \Delta}$.

Let $G=(V,E)$ be an instance of $\mathcal{P}$.
The algorithm $\mathcal{A}'$ on $G$ works as follows. Let $\tau = 1+\log \beta $ be a constant. We use Linial's coloring technique to produce short $\ID$s of length $\ell'$ that are distinct within distance $4f(\Delta) +2\tau+2r$.
Let $G'=(V,E')$ be the graph $E' = \{ \{u,v\} | u,v \in V \textit{ and } \dist_{G}(u,v)\leq 4f(\Delta) +2\tau+2r\}$.
The max degree in $G'$ is at most $\Delta^{4f(\Delta) +2\tau+2r}$ since $\max_{u\in V} \{ |N^{4f(\Delta) +2\tau+2r}(u)|\}\leq \Delta^{4f(\Delta) +2\tau+2r}$.
Each vertex $u\in V$ simulates $G'$ by collecting $N^{4f(\Delta) +2\tau+2r}(u)$ in $O(f(\Delta)+\tau+r)$ time.

We simulate the algorithm of Theorem~\ref{cor:linial-coloring} on $G'$ by treating each of the $\ell$ bit IDs of vertices in $V$ as a color. This produces a $\beta \cdot \Delta^{8f(\Delta) +4\tau+4r} $-coloring, which is equivalent to identifiers of length $\ell' =(8f(\Delta) + 4\tau+4r)\log \Delta + \log \beta$. Although these identifiers are not globally unique, they are distinct in $N^{2f(\Delta) + \tau+r}(u)$ for each vertex $u\in V$.
The time complexity of this process is $$(4f(\Delta) + 2\tau+2r)\cdot O\left(\log^\ast n - \log^\ast \Delta + 1\right).$$

Finally, we apply $\mathcal{A}$ on $G$ while implicitly assuming that the graph size is $2^{\ell'}$ and using the shorter $\ID$s.
The runtime of this execution of $\mathcal{A}$ is:
\begin{align*}
f(\Delta) + \frac{\epsilon \ell'}{\log \Delta} & = f(\Delta) + \frac{\epsilon ((8f(\Delta) + 4\tau+4r)\log \Delta + \log \beta)}{\log \Delta}\\
 & = (1+8\epsilon) f(\Delta) + 1+\frac{\epsilon \log \beta}{\log \Delta} & \epsilon (4\tau+4r)= 1\\
 & \leq (1+8\epsilon) f(\Delta) + \tau  & \log \Delta \geq 1, \epsilon < 1\\
 & \leq 2f(\Delta) + \tau. & 8\epsilon = \frac{2}{\tau+r} \leq 1
 \end{align*}

Whether the output labeling of $u \in V$ is legal depends on the labeling of the vertices in $N^r(u)$, which depends on the graph structure and the $\ID$s in $N^{2f(\Delta) + \tau + r}(u)$.
Due to the hereditary property of the graph class under consideration, for each $u \in V$, $N^{2f(\Delta) + \tau + r}(u)$ is isomorphic to a subgraph of some $2^{\ell'}$-vertex graph in the same class. Moreover, the shortened $\ID$ in $N^{2f(\Delta) + \tau + r}(u)$ are distinct. Therefore, it is guaranteed that the output of the simulation is a legal labeling.

The total time complexity is
$$(4f(\Delta) + 2\tau+2r)\cdot O\left(\log^\ast n - \log^\ast \Delta + 1\right) +  2f(\Delta) +\tau = O\left((1+f(\Delta)(\log^\ast n - \log^\ast \Delta+1) \right).$$
\end{proof}

Combining Theorem~\ref{thm:lb-det-linial} with Corollary~\ref{thm:lb-rnd} and setting $f(\Delta) = O(1)$ provides a new proof of Theorem~\ref{thm:lb-det} for small enough $\Delta$.
To see this, notice that any lower bound for the 
$\RandLOCAL$ model with error probability $1/\poly(n)$ 
can be adapted to $\DetLOCAL$ since we can randomly pick 
$O(\log n)$-bit IDs that are distinct with probability $1-1/\poly(n)$.
From Corollary~\ref{thm:lb-rnd} any $\DetLOCAL$ algorithm that $\Delta$-colors a degree-$\Delta$ tree requires $\Omega ({\log_\Delta \log n})$ time.
However, Theorem~\ref{thm:lb-det-linial} states that any $\DetLOCAL$ algorithm running in $O(1) + o(\log_\Delta n)$ time can be speeded up to run in $O\left(\log^\ast n - \log^\ast \Delta + 1\right)$ time. This contradicts the lower bound whenever $\log_\Delta \log n \gg \log^\ast n - \log^\ast \Delta + 1$.
Hence $\Delta$-coloring a degree-$\Delta$ tree takes $\Omega(\log_\Delta n)$ time in $\DetLOCAL$ for small enough $\Delta$ such that $\log_\Delta \log n \gg \log^\ast n - \log^\ast \Delta + 1$.

Another consequence of Theorem~\ref{thm:lb-rnd} is that the deterministic time complexity of a problem can either be solved very efficiently (i.e. in $O\left((1+f(\Delta)(\log^\ast n - \log^\ast \Delta + 1)\right)$ time) or requires $\Omega(f(\Delta) + {\log_\Delta n})$ time (which is at least the order of the diameter when the underlying graph is a complete regular tree). Such a consequence is the strongest when $\Delta$ is small. For example, if $\Delta$ is a constant, Theorem~\ref{thm:lb-rnd} implies the following corollary:

\begin{corollary} \label{cor-gap}
The time complexity of any LCL problem on any hereditary graph class that has constant $\Delta$ in the $\DetLOCAL$ model is either $\Omega(\log n)$ or $O(\log^\ast n)$.
\end{corollary}

A simple adaptation of the proof of Theorem~\ref{thm:lb-det-linial} shows an even stronger dichotomy when $\Delta=2$.

\begin{theorem}\label{thm:cycle}
The $\DetLOCAL$ time complexity of any LCL problem on any hereditary graph class with $\Delta = 2$ is either $\Omega(n)$ or $O(\log^\ast n)$.
\end{theorem}

We remark that an interpretation of the time complexity requirement in Theorems~\ref{thm:lb-det-linial} and~\ref{thm:cycle} is that the diameter of a graph with maximum degree $\Delta$ is at least $\Omega({\log_\Delta n})$ for $\Delta \geq 3$ and $\Omega(n)$ when $\Delta = 2$. If we allow the possibility for an algorithm to see the entire graph, then the algorithm can solve the problem globally.

Given a $O(\sqrt{\log n})$-time deterministic algorithm, one may feel that it is possible to use Theorem~\ref{thm:lb-det-linial} to improve the time complexity to $O(\log^\ast n)$
since $\sqrt{\log n} = o(\log_\Delta n)$ for the case $\Delta = \exp({o(\sqrt{\log n})})$.
However, the class of graphs with $\Delta = \exp({o(\sqrt{\log n})})$ is not hereditary, and so Theorem~\ref{thm:lb-det-linial} does not apply.
Nonetheless, Linial's coloring technique can be made to speed up algorithms with time complexity of the form $f(\Delta) + g(n)$.

\begin{theorem}\label{thm:gen}
Let $\mathcal{P}$ be an LCL graph problem with parameters $r$, $\Sigma$, and $\mathcal{C}$,  and let $\mathcal{A}$ be a $\DetLOCAL$ algorithm for solving $\mathcal{P}$.
Suppose that the runtime of the algorithm $\mathcal{A}$ on instances of $\mathcal{P}$ from a hereditary graph class
is $O(\log^{k} \Delta + \log^{\frac{k}{k+1}} n)$. Then there exists a deterministic algorithm $\mathcal{A}'$ that solves $\mathcal{P}$ on the same instances in $O(\log^{k} \Delta (\log^\ast n - \log^\ast \Delta + 1))$ time.
\end{theorem}
\begin{proof}
Notice that for any instance of $\mathcal{P}$ with $n$ vertices and ID length $\ell$, it must be that $\ell \geq \log n$ and so the running time of $\mathcal{A}$ on such instances is bounded by $\epsilon_1 \log^{k} \Delta + \epsilon_2 \ell^{\frac{k}{k+1}}$, for some constant $\epsilon_1, \epsilon_2$.

We set $\tau =  \epsilon \log^{k} \Delta$, with the parameter $\epsilon$ to be determined. Similar to the proof of Theorem~\ref{thm:lb-det-linial}, the algorithm $\mathcal{A}'$  first produces shortened $\ID$ that are distinct for  vertices within distance $2\tau+2r$, and then simulates $\mathcal{A}$ on the shortened $\ID$ in $\tau$ rounds.

Let $G'=(V,E')$ be the graph $E' = \{ \{u,v\} | u,v \in V \textit{ and } \dist_{G}(u,v)\leq  2\tau+2r\}$.
The maximum degree in $G'$ is at most $\Delta^{2\tau+2r}$.
Each vertex $u\in V$ simulates $G'$ by collecting $N^{2\tau+2r}(u)$ in $O(\tau+r)$ time.

We simulate the algorithm of Theorem~\ref{cor:linial-coloring} on $G'$ by treating each of the $\ell$ bit IDs of vertices in $V$ as a color. This produces a $\beta \cdot \Delta^{4\tau+4r} $-coloring, which is equivalent to identifiers of length $\ell' =(4\tau+4r)\log \Delta + \log \beta$.
Although these identifiers are not globally unique, they are distinct in $N^{\tau+r}(u)$ for each vertex $u\in V$.
The time complexity of this process is $$(2\tau+2r)\cdot O\left(\log^\ast n - \log^\ast \Delta + 1\right).$$

Finally, we apply $\mathcal{A}$ on $G$ while implicitly assuming that the graph size is $2^{\ell'}$ and using the shorter $\ID$s. By setting $\epsilon$ as a large enough number such that $\epsilon_1 + \epsilon_2 \left(4 (\epsilon +r+\log \beta) \right)^{\frac{k}{k+1}} \leq \epsilon$, the runtime of this execution of $\mathcal{A}$ is:

\begin{align*}
\epsilon_1 \log^{k} \Delta + \epsilon_2 \left(\ell'\right)^{\frac{k}{k+1}}
& = \epsilon_1 \log^{k} \Delta + \epsilon_2 \left( (4\tau+4r)\log \Delta + \log \beta \right)^{\frac{k}{k+1}}  \\
& \leq  \epsilon_1 \log^{k} \Delta + \epsilon_2 \left( 4( \epsilon \log^k \Delta + r + \log \beta )\log \Delta \right)^{\frac{k}{k+1}} \\
& \leq  \epsilon_1 \log^{k} \Delta + \epsilon_2 \left( 4 (\epsilon + r+\log \beta) \log^{k+1} \Delta \right)^{\frac{k}{k+1}} \\
& = \left(\epsilon_1 + \epsilon_2 \left(4 (\epsilon +r+\log \beta) \right)^{\frac{k}{k+1}}\right) \log^{k} \Delta \\
& \leq \epsilon \log^{k} \Delta \\
& = \tau
\end{align*}

Whether the output labeling of $u \in V$ is legal depends on the labeling of the vertices in $N^r(u)$, which depends on the graph structure and the $\ID$s in $N^{\tau + r}(u)$.
Due to the hereditary property of the graph class under consideration, for each $u \in V$, $N^{\tau + r}(u)$ is isomorphic to a subgraph of some $2^{\ell'}$-vertex graph in the same class. Moreover, the shortened $\ID$ in $N^{\tau + r}(u)$ are distinct. Therefore, it is guaranteed that the output of the simulation is a legal labeling.

The total time complexity is at most

$$(2\tau+2r) \cdot O (\log^\ast n - \log^\ast \Delta + 1) + \tau = O(\log^{k} \Delta (\log^\ast n - \log^\ast \Delta + 1)).$$
\end{proof}

\paragraph{A note about MIS lower bounds.} Kuhn,  Moscibroda, and  Wattenhofer~\cite{KuhnMW04} showed that for a variety of problems (including MIS) there is a lower bound of $\min(\log \Delta/\log \log \Delta, \sqrt{\log n/\log \log n})$ rounds. The lower bound graph they used to prove such these result has
$\log \Delta/\log\log \Delta = O( \sqrt{\log n/\log \log n})$. By Theorem~\ref{thm:gen}, setting $k=1$ implies that if there is a deterministic algorithm for MIS that runs in
$O(\sqrt{\log n})$ time, then there is another deterministic algorithm running in $O(\log \Delta (\log^\ast n - \log^\ast \Delta + 1))$ time.
Interestingly, Barenboim, Elkin, Pettie, and Schneider~\cite{BEPS16} showed that an MIS algorithm in $\RandLOCAL$ running in
$O(\log^k \Delta + f(n))$-time implied another $\RandLOCAL$ algorithm running in $O(\log^k \lambda + \log^{1-\frac{1}{k+1}} n + f(n))$ time
on graphs of arboricity $\lambda$.

\section{Algorithms for $\Delta$-coloring Trees \label{sec.ub}}

In Section~\ref{sec.lb}, we showed that the problem of $\Delta$-coloring on trees has an $\Omega({\log_\Delta n})$ deterministic lower bound and an $\Omega({\log_\Delta \log n})$ randomized lower bound. These lower bounds have matching upper bounds by merely a $\log^\ast n$ additive term.

The algorithm of Barenboim and Elkin~\cite{BarenboimE10} demonstrates that the deterministic bound is essentially tight.
They proved that $\Delta$-coloring unoriented trees, where $\Delta\ge 3$, takes $O(\log_\Delta n + \log^* n)$ time.
This is actually a special case of their algorithm, which applies to graphs of bounded arboricity $\lambda$.

\begin{theorem} [\cite{BarenboimE10}] \label{thm:up-det}
For $q \geq 3$, there is a $\DetLOCAL$ algorithm for $q$-coloring trees in $O({\log_q n} + \log^\ast n)$ time,
independent of $\Delta$.
\end{theorem}

Pettie and Su~\cite{PettieS15} gave randomized algorithms for $(4+o(1))\Delta/\ln\Delta$-coloring triangle-free graphs.  Their algorithm makes extensive
use of the distributed \Lovasz{} local lemma~\cite{ChungPS14} and runs in $\Omega(\log n)$ time.  Pettie and Su sketched a proof that $\Delta$-coloring trees
takes $O({\log_\Delta \log n} + \log^\ast n)$ time, at least for sufficiently large $\Delta$.

\begin{theorem} [\cite{PettieS15}] \label{thm:up-rnd}
There exists a large constant $\Delta_0$ such that when $\Delta \geq \Delta_0$, there is a $\RandLOCAL$ algorithm
for $\Delta$-coloring  trees in $O({\log_\Delta \log n} + \log^\ast n)$ time.
\end{theorem}

The nature of the proof of Theorem~\ref{thm:up-rnd} makes it difficult to calculate a specific $\Delta_0$ for which the theorem applies. Moreover, the proof is only sketched, hidden inside more complicated ideas.
We address both of these issues. Firstly, we provide a simple algorithm and elementary proof of Theorem~\ref{thm:up-rnd}. Secondly, we prove Theorem~\ref{thm:coloring-trees-constant-Delta},
which combines Theorem~\ref{thm:up-rnd} with a new technique for constant $\Delta \ge 55$, thereby providing a randomized algorithm for $\Delta$-coloring a tree that runs in $O(\log_\Delta \log n + \log^\ast n)$ time for any constant $\Delta\ge 55$.

\subsection{A simple proof of Theorem~\ref{thm:up-rnd}.}
For a graph $G=(V,E)$ we say that a subset $S \subseteq V$ is a {\em distance-$k$} set if the following two conditions are met:
\begin{enumerate}
\item For any two distinct vertices  $u,v \in S$, we have $u \notin N^{k-1}(v)$.
\item Let $G^k=(V,E^k)$, where there is an edge $(u,v)\in E^k$ if and only if $\dist_G(u,v)= k$. Then $S$ is connected in $G^k$.
\end{enumerate}
We make use of the following lemmas in the proof of Theorem~\ref{thm:up-rnd}.
While the proof of this lemma is implicit in~\cite{BEPS16} we reproduce it here
for sake of clarity.

\begin{lemma}[\cite{BEPS16}] \label{lem:shatter}
The number of distinct distance-$k$ sets of size $t$ is less than $4^t \cdot n \cdot \Delta^{k(t-1)}$.
\end{lemma}
\begin{proof}
A distance-$k$ set is spanned by a tree in $G^k$. There are less than $4^t$ distinct unlabeled trees of $t$ vertices, and there are less than
$n \Delta^{k(t-1)}$ ways to embed a $t$-vertex tree in $G^{k}$. The lemma follows since there is an injective mapping from the family of distance-$k$ sets of size $t$ to
subtrees of $t$ vertices in $G^k$.
\end{proof}

\begin{lemma} [Chernoff bound]\label{lem:chernoff}
Let  $X$ be the sum of $n$ i.i.d. random 0/1 variables. For any $0<\delta<1$, we have:
\begin{align*}
\text{For } 0<\delta<1, & \hspace{0.4cm} \mathrm{Pr}[X \leq (1-\delta)\mathbb{E}[X]] < \exp\left(-\delta^2 \mathbb{E}[X] / 2\right).\\
\text{For } 0<\delta<1, & \hspace{0.4cm} \mathrm{Pr}[X \geq (1+\delta)\mathbb{E}[X]] < \exp\left(-\delta^2 \mathbb{E}[X] / 3\right).\\
\text{For } \delta \geq 1, & \hspace{0.4cm} \mathrm{Pr}[X \geq (1+\delta)\mathbb{E}[X]] < \exp\left(-\delta\mathbb{E}[X] / 3\right).\\
\end{align*}
\end{lemma}

\begin{proof}[Proof of Theorem~\ref{thm:up-rnd}]
Our algorithm has two phases. The first phase, which takes $t=O(\log^\ast \Delta)$ rounds, partially colors the graph using colors in $\{1,2, \ldots, \Delta - \sqrt{\Delta}\}$. The second phase, which takes $O(\log_\Delta \log n + \log^\ast n)$ rounds, applies a deterministic algorithm to $\sqrt{\Delta}$-color the remaining uncolored vertices using  colors in $\{\Delta - \sqrt{\Delta}+1, \ldots, \Delta\}$.
We assume throughout the proof that $\Delta$ is at least a large enough constant.

\paragraph{Phase 1.}
The first phase of the algorithm takes $t = O(\log^\ast \Delta)$  rounds. In each round, the algorithm attempts to color some uncolored vertices. We will explain soon how uncolored vertices decide if they participate in a given round.
In the beginning of round $i$, for each vertex $v\in V$, let $\Psi_i(v)$ denote $v$'s available palette (i.e. the set of colors that $v$ can choose in round $i$), and let $N_i(v)$ denote the set of uncolored vertices adjacent to $v$ that are trying to color themselves in this round.
Initially, we set $N_1(v) = N(v)$, and $\Psi_1(v)=\{1,2, \ldots, \Delta - \sqrt{\Delta}\}$, for all $v$.
That is, in the first round all vertices attempt to color themselves, and they all have the full palette of this phase available for choices of a color.

We maintain the following two properties for each vertex $v$ that is attempting to color itself at round $i$:
\begin{itemize}
\item {\bf Large Palette Property.} $\mathcal{P}_1(v): |\Psi_{i}(v)| \geq \frac{\Delta}{200}$.
\item {\bf Small Degree Property.} $\mathcal{P}_2(v): |N_i(v)| \leq \frac{\Delta}{c_i}$, where $c_i$ is defined as: $c_1 = 1$, $c_2 = 1-\frac{1}{200}$, and $c_{i} = \min\left\{ \Delta^{0.1}, \; c_{i-1} \cdot \exp\paren{\frac{c_{i-1}}{3\cdot 200 \cdot e^{200}}}\right\}$ for $i > 2$.
\end{itemize}
Notice that $c_i$ is a constant, for all $i$.
Let $t$ be the smallest number $i$ such that $c_i = \Delta^{0.1}$. Notice that $t=O(\log^\ast \Delta)$.

The intuition behind the two properties $\mathcal{P}_1(v)$ and $ \mathcal{P}_2(v)$ is that they ensure that (i)  participating vertices  always have a large enough palette to use, and (ii) there is a large separation between the palette size and the degree so that we can color a large fraction of vertices in each round.

For each $1\leq i \leq t$, the $i^{\text{th}}$ round consists of two constant time sub-routines {\sf ColorBidding($i$)} and {\sf Filtering($i$)}.
In {\sf ColorBidding($i$)}, each participating vertex $v$ selects a random subset of colors $S_v$.
If there is a color in $S_v$ that does not belong to $\bigcup_{u \in N_i(v)} S_u$, the vertex $v$ {\em succeeds} and colors itself with any such color. If such a color is chosen, denote it by $\mathrm{Color} (v) $.
After {\sf ColorBidding($i$)}, we execute {\sf Filtering($i$)} which filters out some vertices and thereby  preventing $\mathcal{P}_1$ and $\mathcal{P}_2$ from being violated.
Such vertices are called {\em bad} vertices, and they will no longer participate in the remaining rounds of Phase 1.

\begin{framed}
{\sf ColorBidding($i$).}

Do the following steps in parallel for each uncolored vertex $v$ that is not bad:
\begin{enumerate}
\item If $c_i = 1$, then choose one color uniformly at random from $\Psi_i(v)$, and $S_{v}$ contains only this color.
    Otherwise ($c_i > 1$), construct the set $S_{v}$ by independently including each color of $\Psi_i(v)$ with probability $\frac{c_i}{|\Psi_i(v)|}$.

\item If $S_{v} \setminus \bigcup_{u \in N_i(v)} S_u \neq \emptyset$, then permanently color $v$ by picking an arbitrary color in $S_{v} \setminus \bigcup_{u \in N_i(v)} S_u$ for $\mathrm{Color} (v)$.
\item $\Psi_{i+1}(v) \leftarrow \Psi_{i}(v) \setminus \{\mathrm{Color}(u) \;|\; u \in N_i(v) \text{ is permanently colored}\}$.
\end{enumerate}
\end{framed}

We define $N_i'(v)$ as the set of participating vertices {\em after} {\sf ColorBidding($i-1$)} and {\em before} {\sf Filtering($i-1$)} that are adjacent to $v$. In other words, $N_i'(v) = N_{i-1}(v) \setminus \{u | u$ is permanently colored in {\sf ColorBidding($i-1$)}$\}$.

\begin{framed}
{\sf Filtering($i$).}

For each uncolored vertex $v$ that is not bad:
\begin{enumerate}
\item If $i = 1$ and $|\Psi_2(v)| - |N_2'(v)| < \frac{\Delta}{200}$, then mark $v$ as a bad vertex.
\item If $1 < i < t$ and $|N_{i+1}'(v)| > \frac{\Delta}{c_{i+1}}$, then mark $v$ as a bad vertex.
\item If $i = t$ then mark $v$ as a bad vertex.
\end{enumerate}
\end{framed}

\paragraph{Phase 2.} By the filtering rule for $i = t$, all the remaining uncolored vertices after the Phase 1 are bad vertices. We color the bad vertices in Phase 2.
We will later prove  that after phase 1, with high probability a connected component induced by bad vertices has size at most $\Delta^4 \log n$.
Hence we use Theorem~\ref{thm:up-det} to $\sqrt{\Delta}$-color such connected components using the $\sqrt{\Delta}$ reserved colors. For simplicity, if this phase lasts for too long (which may happen with low probability) the algorithm just stops and fails.

\paragraph{Runtime.}
The runtime of phase 1 is $t=O(\log^\ast \Delta)$ rounds.
The runtime of phase 2 is
$O\big({\log_{\sqrt{\Delta}}\left(\Delta^{4}\log n\right)}$ + $\log^\ast\left(\Delta^{4}\log n\right)\big) = O\left({\log_\Delta \log n} + \log^\ast n\right).$
Thus, the total runtime is $O\left({\log_\Delta \log n} + \log^\ast n\right)$ rounds.

\paragraph{Analysis.}
The analysis of phase 2 relies only on proving that, with high probability, all of the connected components induced by bad vertices after phase 1 are of size at most $\Delta^4 \log n$. So we focus on analyzing phase 1.

The decision of whether a vertex $v$ that participates in round $i$ becomes marked as bad in this round depends on the vertices in $N^2(v)$ that participate in this round and the random bits used by these vertices during this round. Our analysis will work for any such arbitrary setting (of participating vertices in $N^2(v)$ and their random bits).
In particular, we prove that in any round of Phase 1, for any vertex $v$ participating in this round and any arbitrary choice of vertices in $N^2(v)$ that participate in this round with their random bits (that are used in this round), $v$ becomes a bad vertex with probability at most $\exp(-\text{poly}(\Delta)$. Such a proof means that the choice of whether $v$ becomes bad or not does not depend on any arbitrary behaviour of vertices participating in this round that are not in $N^2(v)$.
This proof is covered by the following claims, whose proofs are given in the Appendix (we partition the cases since each case requires a different proof).
\begin{claim}\label{clm-1}
The probability that a vertex $v$ is marked as bad in round $i=1$ is at most $\exp(-\Omega(\Delta))$, regardless of the random bits used by vertices in $N^2(v)$.
\end{claim}
\begin{claim}\label{clm-2}
The probability that a vertex $v$ that participates in round $1<i<t$ is marked as bad in round $i$ is at most $\exp(-\Omega(\Delta^{0.1}))$, regardless of the random bits used by vertices in $N^2(v)$.
\end{claim}
\begin{claim}\label{clm-3}
The probability that a vertex $v$ that participates in round $i=t$ is marked as bad in round $i$ is at most $\exp(-\Omega(\Delta^{0.1}))$, regardless of the random bits used by vertices in $N^2(v)$.
\end{claim}

By the union bound for all rounds in Phase 1, the probability that any vertex $v$ becomes a bad vertex after Phase 1 is $O(\log^\ast \Delta) \exp(-\text{poly}(\Delta)) = \exp(-\text{poly}(\Delta))$,
regardless of the choice of random bits for all vertices not in $N^2(v)$.
Therefore, just before Phase 2, for any distant-$5$ set $T$ of size $s$, the probability that all vertices in $T$ are bad is at most $\exp(-s \cdot\text{poly}(\Delta))$.
By Lemma~\ref{lem:shatter}, there are at most $4^s \cdot n \cdot \Delta^{4(s-1)}$ distinct distant-$5$ set $T$ of size $s$.
By the union bound, with probability at least $\left( 4^s \cdot n \cdot \Delta^{4(s-1)}\right) \cdot \exp(-s \cdot\text{poly}(\Delta))$, there is no distant-$5$ set of size $s$ that contains only bad vertices.  This probability can be upper bounded by $n^{-c}$ for any $c$ when $s = \log n$.

This concludes the proof of Theorem~\ref{thm:up-rnd}.
\end{proof}

\subsection{Algorithm for $\Delta \ge 55$.}
Inherently, the above proof  (and also the proof in \cite{PettieS15}) is hard to analyze quantitatively without the aid of $O(\cdot)$ notation. It seems to require a very large $\Delta$   for the proof to work well, since in each round several Chernoff bounds are applied to make sure that some requirements are met (see the proof of Claims~\ref{clm-1},~\ref{clm-2}, and~\ref{clm-3} in the Appendix), and we need a large enough $\Delta$ to make these Chernoff bounds work.
In what follows we present a different algorithm with a significantly simpler analysis for $\Delta$-coloring trees.

\begin{theorem}\label{thm:coloring-trees-constant-Delta}
For $\Delta \geq 55$, there exists a $\RandLOCAL$ algorithm $\Delta$-coloring of a tree can be computed in $O({\log_\Delta \log n} + \log^\ast n)$ time.
\end{theorem}
\begin{proof}
We assume that $\Delta = O(1)$ is constant.  If it is sufficiently large, apply Theorem~\ref{thm:up-rnd}.
Our algorithm has three phases:

\paragraph{Phase 1.} We execute the following procedure to partially color the graph with colors in $\{4, 5, \ldots, \Delta\}$.

\begin{framed}
\begin{enumerate}
\item[] Initially $U \leftarrow V$.
\item[] For $i$ from $\Delta$ downto $4$, do the following steps in parallel for each vertex $v \in U$:
\item \label{s1} Choose a real number $x(v) \in [0,1]$ uniformly at random.
\item \label{s2} Let $\displaystyle K = \left\{v \;|\; x(v) < \min_{u \in N(v) \cap U} x(u)\right\}$ be the set of all vertices holding local minima.
\item \label{s3} Find any MIS $I \supseteq K$ of $U$. All vertices in $I$ are colored $i$.
\item \label{s4} Set $U\leftarrow U \setminus I$  (remove all colored vertices).
\end{enumerate}
\end{framed}

The above procedure ensures that  the number of uncolored neighbors of a vertex $v \in U$ is at most $i-1$ after step~\ref{s4}.
Therefore, at the end of the Phase 1, we have $|N(v) \cap U| \leq 3$ for any uncolored vertex $v$.

The MIS required in Step~\ref{s3} can be computed in $O(\Delta + \log^\ast n) = O(\log^\ast n)$ time \cite{BarenboimEK14}, or in $O(\Delta^2 + \log^* n) = O(\log^* n)$ time via Theorem~\ref{cor:linial-coloring}.

\paragraph{Phase 2.} We will later show that the set of the vertices $S = \{v \in U \text{ s.t. } |N(v) \cap U| = 3\}$ form connected components of size at most $O(\log n)$ with probability $\geq 1 - n^{-c}$. Hence we  apply Theorem~\ref{thm:up-det} to 3-color the set $S$ (using the colors $1,2,3$) in $O(\log \log n)$ time. We then update $U = U \setminus S$ after coloring the vertices in $S$.

\paragraph{Phase 3.} For each vertex $v$ that remains uncolored, the number of its available colors (i.e. $\{1, \ldots, \Delta\} \setminus \{\text{color}(u) \;|\; u \in N(v) \text{ is colored}\}$) is greater than the number of its uncolored neighbors (i.e. $|N(v) \cap U|$). We apply an $O(\log^\ast n)$-time MIS algorithm twice to get a 3-coloring of vertices in $U$ (with three colors $1',2',3'$). For $i = 1', 2', 3'$, we recolor the vertices that are colored with $i'$ using any of its available colors.

\medskip

In view of the above, to prove the theorem, it suffices to show that the set  $S = \{v \in U \text{ s.t. }  |N(v) \cap U| = 3\}$ (which is defined in Phase 2) form connected components of size at most $O(\log n)$ with probability $\geq 1 - n^{-c}$.

Given any distant-$3$ set $T$ of size $t$, we select any vertex in $V \setminus T$ as a root to make the tree rooted, and for each $v_i \in T$, we define $w_i$ as the parent of $v_i$. We also define $D_i = \bigcup_{u \in N(v_i) \setminus \{w_i\}} N(u)$, for each $v_i \in T$. We observe that $\forall v_i, v_j \in T, D_i \cap D_j = \emptyset$. We prove the the following claim (The proof is given in the Appendix):
\begin{claim}\label{clm-4}
There is some constant $\epsilon$ such that the probability for a vertex $v_i$ to be in $S$ conditioned on arbitrary behavior of vertices not in $D_i$ is at most $\frac{1 - \epsilon}{4\Delta^{3}}$.
\end{claim}
%Our plan is to bound the probability for a vertex $v$ to be in $S$ by $\frac{1 - \epsilon}{4\Delta^{3}}$, for any some constant $\epsilon$, conditioned on arbitrary behavior of vertices not in $D_i$.
Thus, the probability that all vertices in $T$ belong to $S$ is at most $\left(\frac{1 - \epsilon}{4\Delta^{3}}\right)^t$.
By Lemma~\ref{lem:shatter}, as long as
$$t \geq \frac{(c+1) \log n  - \log (\Delta^3)}{\log \frac{1}{1-\epsilon}},$$
with probability $\geq 1 - n^{-c}$ we have:
$$4^t n \Delta^{3(t-1)} \cdot \left(\frac{1 - \epsilon}{4\Delta^{3}}\right)^t \leq n^{-c}.$$
In words, with high probability there is no distant-$3$ set of size $t$ whose vertices are all in $S$.

Since (for $\Delta \geq 2$) any connected subgraph with  number of vertices at least $ \Delta^2 t$ must contain a distant-$3$ set of size $t$, we conclude that, with high probability, $S$ forms connected components of size at most $O(\Delta^2\log n) = O(\log n)$.
\end{proof}

\bibliographystyle{plain}
\bibliography{references}

\newpage

\appendix
\begin{Large}
\noindent {\bf Appendix} \\
\end{Large}

\section{Proofs of Claims~\ref{clm-1},~\ref{clm-2}, and~\ref{clm-3}}\label{sect:mmm}

The proofs for these claims rely on the Large Palette Property $\mathcal{P}_1(v)$ and the Small Degree Property $\mathcal{P}_2(v)$. First, we observe that our filtering rules imply that these two properties hold after each round:
\begin{itemize}
\item The filtering rule for $i=1$ guarantees that the Large Palette Property $\mathcal{P}_1(v)$ is met for all vertices that remain after the filtering. Notice that $|\Psi_2(v)| - |N_2'(v)| \geq \frac{\Delta}{200}$ implies $|\Psi_i(v)| \geq \frac{\Delta}{200}$, for all $i$.
\item The filtering rules for $i=1$ and $1 < i < t$ ensures that the Small Degree Property $\mathcal{P}_2(v)$ holds for all $i$, since $|N_{i+1}'(v)| \leq \frac{\Delta}{c_{i+1}}$ implies that $|N_{i+1}(v)| \leq \frac{\Delta}{c_{i+1}}$.
\end{itemize}
%The last filtering rule, which marks all uncolored vertices after {\sf ColorBidding($t$)}, implies that  all vertices are colored except the bad vertices.

\begin{proof}[Proof of Claim~\ref{clm-1}]
Recall that a vertex $v$ is marked as bad in round 1 only if $|\Psi_2(v)| - |N_2'(v)| < \frac{\Delta}{200}$.
Thus, we assume $|N(v)|\ge \frac{199\cdot \Delta}{200}$, since otherwise $v$ will not be marked as bad.

For each neighbor $u$ of $v$, let $E_u$ denote the event of $u$ being colored in the first round. Since the graph is a tree, for all $u \in N(v)$ the events $E_u$ are independent. Assuming sufficiently large $\Delta$, we have:
$$\mathrm{Pr}[E_u] \geq \left(1-\frac{1}{\Delta- \sqrt{\Delta}}\right)^{|N(u)|} \geq \left(1-\frac{1}{\Delta - \sqrt{\Delta}}\right)^\Delta \geq \frac{1}{3}.$$

By a Chernoff bound (with $\delta = \frac{79}{199}$, and the expected number of colored neighbors being at least $\frac{199}{200} \cdot \frac{\Delta}{3}$), the number of colored neighbors of $v$ in the first round is at least $\frac{\Delta}{5} = (1-\delta) \cdot  \frac{199}{200} \cdot \frac{\Delta}{3}$  with probability at least
$$1-\exp\left(\frac{-\left(\frac{79}{199}\right)^2\cdot \left( \frac{199}{200} \cdot \frac{\Delta}{3} \right)}{2}\right).$$

Let $S=\{u_1, u_2, \ldots\}$ be the colored subset of $N(v)$. In what follows, we assume that $|S|\geq \frac{\Delta}{5}$. Conditioned on the color selected by $v$ and the event that $S$ is the colored subset of $N(v)$, each $u_j$ independently selects a color uniformly at random from $\{1,\ldots, \Delta - \sqrt{\Delta}\} \setminus \{\mathrm{Color}(v)\}$.

If $\frac{\Delta}{10} - |\bigcup_{j=1}^{\Delta/10} \{ \mathrm{Color}(u_j)\}| \geq \frac{\Delta}{200}$, then $|\Psi_2(v)| - |N_2(v)| \geq \frac{\Delta}{200}$, and $v$ is not marked as bad. Otherwise,  each $u_j$, $j > \frac{\Delta}{10}$, chooses a color that is already chosen by some $u_k$, $k<j$, with probability at least
$$\frac{\frac{\Delta}{10} - \frac{\Delta}{200}}{\Delta - \sqrt{\Delta} - 1} \geq \frac{1}{11}.$$

As a result, the expected value of $|S| - |\bigcup_{j=1}^{|S|} \{\mathrm{Color}(u_j)\}|$ is at least $\frac{\Delta}{10} \cdot \frac{1}{11} = \frac{\Delta}{110}$. By a Chernoff bound (with $\delta = \frac{9}{20}$), this value is at least $\frac{\Delta}{200} = (1 -\delta) \cdot \frac{\Delta}{110}$ with  probability at least
$$1-\exp\left(\frac{-\left(\frac{9}{20}\right)^2 \cdot \frac{\Delta}{110}}{2}\right) = 1 - \exp\left(\frac{-162\Delta}{11}\right).$$

By definition, $|\Psi_2(v)| = |\Psi_1(v)| - |\bigcup_{j=1}^{|S|} \{\mathrm{Color}(u_i)\}|$, and $|N_2'(v)| = |N_1(v) \setminus S| =|N_1(v)| - |S|$.
In view of the above, $|\Psi_2(v)| - |N_2(v)| \geq \frac{\Delta}{200}$ (and $v$ is not marked as bad) with  probability at least $$1 - \exp\left(\frac{-\left(\frac{79}{199}\right)^2\cdot \left( \frac{199}{200} \cdot \frac{\Delta}{3} \right)}{2}\right)  - \exp\left(\frac{-162\Delta}{11}\right) = 1 - \exp(-\Omega(\Delta)).$$
\end{proof}

\begin{proof}[Proof of Claim~\ref{clm-2}]
Our goal is to show that $|N_{i+1}'(v)| \leq \frac{\Delta}{c_i \cdot \exp(\frac{c_{i}}{3\cdot 200 \cdot e^{200}})} = \frac{\Delta}{c_{i+1}}$ with high probability.

We include each available color in $S_v$ independently with probability $\frac{c_i}{|\Psi_i(u)|}$, so the expected value of $|S_v|$ is $c_i$. By a Chernoff bound, the event of $|S_v|\leq \frac{\Delta}{2\cdot 200} = (1+\delta)c_i$, where $\delta = \left(\frac{\Delta}{2\cdot 200 \cdot c_i} - 1\right)$, happens with  probability at least:
$$1 - \exp\left(\frac{-\left(\frac{\Delta}{2\cdot 200 \cdot c_i} - 1\right) \cdot c_i}{2}\right) \geq 1 - \exp\left(\frac{-\Delta}{5\cdot 200}\right).$$

For any $u \in N_i(v)$, a color in $\Psi_i(u) \setminus S_v$ belongs to $S_{u} \setminus \bigcup_{w \in N_i(u)} S_w$ with probability at least
$$\frac{c_i}{|\Psi_i(u)|}  \cdot
\prod_{w \in N_i(u) \setminus \{v\}}\left(1-\frac{c_i}{|\Psi_i(w)|}\right)
\geq
\frac{c_i}{|\Psi_i(u)|} \cdot
\left(1 - \frac{c_i \cdot 200}{\Delta}\right)^{\frac{\Delta}{c_i} - 1}
\geq
\frac{c_i }{ e^{200} |\Psi_i(u)|}.$$
Notice that the term $\frac{c_i}{|\Psi_i(u)|}$ is the probability that a color in $\Psi_i(u) \setminus S_v$ is chosen to be in $S_{u}$, and the term $\prod_{w \in N_i(u) \setminus \{v\}}\left(1-\frac{c_i}{|\Psi_i(w)|}\right)$ is the probability that a color in $\Psi_i(u) \setminus S_v$ does not belong to $\bigcup_{w \in N_i(u)} S_w$.

Under the condition that $|S_v|\leq \frac{\Delta}{2\cdot 200 }$, we have $|\Psi_i(u) \setminus S_v| \geq \frac{\Delta}{200} -  \frac{\Delta}{2 \cdot 200} = \frac{\Delta}{2 \cdot 200}$. Then the set $S_{u} \setminus \bigcup_{w \in N_i(u)} S_w$ is empty with probability at most:
 $$\left(1 - \frac{c_i }{e^{200} |\Psi_i(u)|}\right)^{\frac{\Delta}{2 \cdot 200}}  \leq
\left(1 - \frac{c_i }{ e^{200} {\Delta} }\right)^{\frac{\Delta}{2 \cdot 200}}
\leq \exp\left(-\frac{c_i}{2\cdot 200\cdot e^{200}}\right).$$
Hence  $u \in N_i(v)$ remains uncolored with probability at most $\exp(-\frac{c_i}{2\cdot 200\cdot e^{200}})$.

\medskip

\noindent {\bf [Case 1. $\exp(-\frac{c_i}{2\cdot 200\cdot e^{200}}) \cdot |N_i(v)| \geq \Delta^{0.1}$]} We choose $\delta$ to be the number such that $(1+\delta) \exp(-\frac{c_i}{2\cdot 200\cdot e^{200}}) = -\frac{c_i}{3\cdot 200\cdot e^{200}}$. Notice that we always have $\delta \geq \exp(\frac{1}{6\cdot 200\cdot e^{200}}) - 1$, which is a positive constant. By a Chernoff bound on all vertices in $N_i(v)$, we have $|N_{i+1}'(v)| \leq  \exp(-\frac{c_i}{3\cdot 200\cdot e^{200}}) \cdot |N_i(v)|$ with probability at least
$$1-\exp\left(\frac{-\max\{1, \delta^2\}\cdot{\exp(-\frac{c_i}{2\cdot 200\cdot e^{200}}) \cdot |N_i(v)|}}{3}\right) \geq 1 - \exp(-\Omega(\Delta^{0.1})).$$

\medskip

\noindent {\bf [Case 2.  $\exp(-\frac{c_i}{2\cdot 200\cdot e^{200}}) \cdot |N_i(v)| < \Delta^{0.1}$]} By a Chernoff bound on all vertices in $N_i(v)$, we have $|N_{i+1}'(v)| \leq \Delta^{0.8} \cdot \exp(-\frac{c_i}{2\cdot 200\cdot e^{200}}) \cdot |N_i(v)| \leq \Delta^{0.9}$ with probability at least
$$1-\exp\left(\frac{(\Delta^{0.8} - 1) \cdot \exp(-\frac{c_i}{2\cdot 200\cdot e^{200}}) \cdot |N_i(v)|}{3}\right) \geq 1-\exp\left(\frac{-\Delta^{0.9} + \Delta^{0.1}}{3}\right).$$

In any case, we have $|N_{i+1}'(v)| \leq \max \left\{\frac{|N_i(v)|}{\exp\left({\frac{c_i}{3\cdot 200\cdot e^{200}}}\right)}, \Delta^{0.9}\right\} = \frac{\Delta}{c_{i+1}}$ (and so $v$ is not marked as bad) with probability at least  $1-\exp(-\Omega(\Delta^{0.1}))$.
\end{proof}

\begin{proof}[Proof of Claim~\ref{clm-3}]
A color in $\Psi_i(v)$ belongs to $S_{v} \setminus \bigcup_{u \in N_i(v)} S_u$ with probability at least
$$\frac{c_i}{|\Psi_i(v)|} \cdot \prod_{u \in N_i(v) }\left(1-\frac{c_i}{|\Psi_i(u)|}\right)
\geq \frac{c_i}{|\Psi_i(v)|} \cdot \left(1 - \frac{c_i 200}{\Delta}\right)^{\frac{\Delta}{c_i}}
\geq \frac{c_i}{1.1  e^{200} |\Psi_i(v)|}.$$

Therefore, $v$ remains uncolored (and is marked as bad) with probability at most
$$\left(1 - \frac{c_i }{1.1 e^{200} |\Psi_i(v)|}\right)^{|\Psi_i(v)|} \leq  \exp\paren{-\frac{ \Delta^{0.1}}{1.1 e^{200}}}.$$
\end{proof}

\section{Proof of Claim~\ref{clm-4}}
\begin{proof}[Proof of Claim~\ref{clm-4}]
For notational simplicity, we write $v \bydef v_i$ and $w \bydef w_i$.

We observe that $v \in S$ only when for all $i=\Delta$ to $4$, at most one neighbor of $v$ is colored with $i$ in Step~\ref{s2}. In addition, in the beginning of Step~\ref{s1}, we must have $|N(v) \cap U| = i$, for all $i=\Delta$ to $4$.

Now, assume that we are at the beginning of Step~\ref{s1}, the vertex $v$ still has no neighbors of repeated colors, and $|N(v) \cap U| = i$. The probability that a neighbor $u \in N(v) \setminus \{w\}$  is colored $i$ in Step~\ref{s2} when $x(v) = z \in [0,1]$ is at least:
$$\int_{y=0}^{y=z} \mathrm{Pr}[\forall r \in (N(u) \cap U) \setminus \{v\}, x(r) \leq y] dy \geq \int_{y=0}^{y=z} (1-y)^{i-1} dy = \frac{1-(1-z)^i}{i}.$$

Note that the variable $y$ represents the random variable $x(u)$, and $u$ is colored when (i) $y \in [0,z)$, and (ii) $x(r) \in (y,1]$ for all $r \in (N(u) \cap U) \setminus \{v\}$. Also note that $|(N(u) \cap U) \setminus \{v\}| \leq i-1$.

We write $p_i(z) = \frac{1-(1-z)^i}{i}$. Then the probability that  at most one neighbor of $v$ is colored with $i$ in Step~\ref{s2} can be upper bounded by:
\begin{align*}
P_{i} &= \int_{x=0}^{x=1} \mathrm{Pr}[\mathrm{binom}(|(N(v) \cap U) \setminus \{w\}|,p_i(x))\leq 1] dx\\
	&\leq \int_{x=0}^{x=1}\left[(1-p_i(x))^{i-1} + (i-1)p_i(x)(1-p_i(x))^{i-2}\right]dx.
\end{align*}
Notice that $|(N(v) \cap U) \setminus \{w\}| \geq i-1$.

When $i=4$, the $P_4$ is about $0.88718$.  As $i$ increases, $P_i$ decreases monotonically, approaching $0.73576$.
By a numerical calculation, so long as $\Delta \geq 55$, the probability that $v$ is in $S$ conditioned on arbitrary behavior of vertices not in $\bigcup_{u \in N(v) \setminus \{w\}} N(u)$ is at most
$$\prod_{i=4}^{i=\Delta} P_i < \frac{1}{4\Delta^{3}},$$
as desired.
\end{proof}

\end{document}